\documentclass[12pt,reqno,twoside]{article}

\usepackage[T1]{fontenc}

\usepackage[]{amsmath, amssymb, amsthm, tabularx}

\usepackage[]{a4, xcolor, here}

\usepackage[]{pstricks, pst-text, pst-node, pst-tree, gastex}

\usepackage{tikz}

\usepackage[tikz]{bclogo}

\usepackage{mhsetup, mathtools}

\usepackage{boites,graphicx}

\usepackage{soul}

\definecolor{pastelyellow}{rgb}{0.99, 0.99, 0.59}
\definecolor{aqua}{rgb}{0.0, 1.0, 1.0} 
\definecolor{aquamarine}{rgb}{0.5, 1.0, 0.83} 
\definecolor{bananayellow}{rgb}{1.0, 0.88, 0.21}
\definecolor{burgundy}{rgb}{0.5, 0.0, 0.13}
\definecolor{ao(english)}{rgb}{0.0, 0.5, 0.0}

\setul{}{0.2ex}
\setulcolor{bananayellow}

\usepackage[normalem]{ulem}

\usepackage{mathrsfs}

\usepackage{stmaryrd}

\usepackage{enumerate}

\usepackage{paralist}

\usepackage{multienum}

\usepackage{multicol}

\usepackage{fancyhdr}

\usepackage{datetime}

\usepackage[multiple]{footmisc}

\usepackage{everypage}
\usepackage[contents={},opacity=1,scale=1.6,
color=gray!90]{background}
\usepackage{ifthen}

\usepackage[]{hyperref} 

\hypersetup{
		colorlinks = true,
		linkcolor = red,
		anchorcolor = black,
		citecolor = blue, 
		filecolor = cyan,
		menucolor = red,
		runcolor = cyan,
		urlcolor = magenta,
		linkbordercolor = {white},
		linktocpage = true
}

\newtheorem{theorem}{Theorem}[section]
\newtheorem{proposition}[theorem]{Proposition}
\newtheorem{lemma}[theorem]{Lemma}

\theoremstyle{definition}
\newtheorem{definition}[theorem]{Definition}

\newtheorem{example}[theorem]{Example}
\newtheorem{remark}[theorem]{Remark}

\makeatletter
\def\thmhead@plain#1#2#3{%
  \thmname{#1}\thmnumber{\@ifnotempty{#1}{ }\@upn{#2}}%
  \thmnote{ {\the\thm@notefont#3}}}
\let\thmhead\thmhead@plain
\makeatother

\newcommand{\cC}{\mathcal{C}}

\newcommand{\cF}{\mathcal{F}}
\newcommand{\cG}{\mathcal{G}}

\newcommand{\cS}{\mathcal{S}}
\newcommand{\cU}{\mathcal{U}}
\newcommand{\cV}{\mathcal{V}}
\newcommand{\cW}{\mathcal{W}}
\newcommand{\cX}{\mathcal{X}}

\newcommand{\rsp}[1]{{\mathrm{rowsp}{#1}}}

\newcommand{\bbF}{{\mathbb F}} 

\renewcommand{\geq}{\geqslant}
\renewcommand{\leq}{\leqslant}

{\end{list}}%

\begin{document}

\renewcommand{\headrulewidth}{0pt}

\rhead{ }
\chead{\scriptsize Flag Codes from Planar Spreads in Network Coding}
\lhead{ }

\title{Flag Codes from Planar Spreads \\ in Network Coding
\renewcommand\thefootnote{\arabic{footnote}}
}

\author{\renewcommand\thefootnote{\arabic{footnote}}
Clementa Alonso-Gonz\'alez\footnotemark[1],\, Miguel \'Angel Navarro-P\'erez\footnotemark[1], \\  
\renewcommand\thefootnote{\arabic{footnote}}
 Xaro Soler-Escriv\`a\footnotemark[1]}

\footnotetext[1]{Dpt.\ de Matem\`atiques, Universitat d'Alacant, 
Sant Vicent del Raspeig, Ap.\ Correus 99, E -- 03080 Alacant. \\ E-mail adresses: \texttt{clementa.alonso@ua.es, miguelangel.np@ua.es, xaro.soler@ua.es.}}


\maketitle

\date

\begin{abstract}
In this paper we study a class of multishot network codes given by families of nested subspaces (flags) of a vector space $\bbF_q^n$, being $q$ a prime power and $\bbF_q$ the finite field of $q$ elements. In particular, we focus on flag codes having maximum minimum distance (\emph{optimum distance flag codes}). We explore the existence of these codes from spreads, based on the good properties of the latter ones. For $n=2k$, we show that optimum distance full flag codes with the largest size are exactly those that can be constructed from a planar spread. We give a precise construction of them as well as a decoding algorithm.
\end{abstract}

\textbf{Keywords:} Network coding, subspace codes, projective space, spreads, flag codes.


\section{Introduction}\label{sec:Introduction}

The concept of {\em Network Coding}, first introduced in \cite{AhlsCai00}, describes a method for attaining a maximum information flow within a network that is an acyclic directed graph with possibly several sources and sinks. It was proved  in \cite{AhlsCai00} that the information rate of a network can be improved by using coding at the nodes of the network, instead of simply routing the received inputs. An algebraic approach to coding in non-coherent networks, called {\em random network coding}, was given by Koetter and Kschischang in \cite{KoetKschi08}. Given a finite field $\bbF_q$, the authors defined the {\em subspace channel} as a discrete memoryless channel with input and output alphabets given by the collection of all possible vector subspaces of $\bbF_q^n$, that is, $\mathcal{P}_q(n)$. The source node transmits an input subspace vector, which is processed in the intermediate nodes and, in the absence of errors, the sink nodes receive the same subspace. In order to correct possible errors or erasures that may happen during the transmission, one can limit the choice of input subspaces to a particular subset of the projective space called {\em subspace code} \cite{KoetKschi08}. The study of subspace codes has led to many papers in recent years (see for instance \cite{TrautRosen18} and the references therein). Most of these articles focus on {\em constant dimension codes}, that is, subspace codes in which all subspaces have the same dimension. An important family of constant dimension codes is the one of {\em spread codes}, which have maximal minimum distance and achieve the largest possible size \cite{MangaGorlaRosen08, MangaTraut14}.

When we use the subspace channel more than once, we talk about {\em multishot subspace codes}. In this kind of codes, introduced in \cite{NobUcho09}, the subspace channel is used many times, in order to transmit sequences of subspaces. As it was explained in that paper, multishot subspace codes appear as an interesting alternative to subspace codes ({\em one-shot subspace codes}) when the field size $q$ or the packet size $n$ can not be increased. Moreover, even if these parameters are acceptable for our communication channel, multishot subspace codes can be useful for solving complexity problems (one-shot codes in $\mathcal{P}_q(nr)$ can be considerably more complicated than {\em $r$-shot} codes over $\mathcal{P}_q(n)$ \cite{NobUcho09}). 
As a particular case, here we consider multishot subspace codes given by sequences of nested subspaces, that is, {\em flag codes}. In the network coding setting, flag codes were first introduced in \cite{LiebNebeVaz18}, as a generalization of constant dimension codes, and a  model of network channel for flags was given. 

In this paper we characterize flag codes having maximum distance  in terms of the constant dimension codes used at each shot. We call them {\em optimum distance flag codes}. Motivated by the good properties of spreads as constant dimension codes, we focus on flag codes than can be constructed from some spread. Moreover, for $n=2k$ we prove that the optimum distance full flag codes with the best possible size are exactly the ones having a {\em planar spread} as the constant dimension code used at the $k$-th shot. We give a precise construction of such flag codes together with a decoding algorithm on the erasure channel. 

It was suggested in \cite{NobUcho10}  that the existence of dependencies among the transmitted subspaces in each sequence may improve the error-correction capability of a multishot code. We prove that this is true for the flag codes given by our construction. 

The structure of the paper is as follows. In Section 2 we give some basic background on constant dimension codes (mainly equidistant codes, partial spreads and spreads) together with some notions on multishot codes. Section 3 is devoted to the study of flag codes focusing on their distance properties. We provide a bound for optimum distance flag codes and we give a characterization of them. In Section 4 we explore how to get optimum distance flag codes with the largest possible size. We discuss if it is possible to get optimum distance flag codes from  $k$-spreads (for some divisor $k$ of $n$). We conclude that for full flag codes this is only possible starting from planar spreads. Next we give a concrete construction of an optimum distance full flag code with the largest size in $\bbF_q^{2k}$ and we develop a decoding process for our code on the erasure channel. We finish Section 4 with some constructions of general type flag codes closely related to our construction for full flag codes. 



\section{Preliminaries}\label{sec:Preliminaries}

Let $q$ be a prime power and $\bbF_q$ the field with $q$ elements. Fix an integer $n>1$ and consider $\mathcal{P}_q(n)$ the set of all vector subspaces of $\bbF_q^n$, i.e., the {\em projective geometry} of $\bbF_q^n$. The natural measure of distance in  $\mathcal{P}_q(n)$ is given by
$$
d_S(\mathcal{U}, \mathcal{V})= \dim(\mathcal{U}+\mathcal{V})-\dim(\mathcal{U} \cap \mathcal{V})
$$
for all $\mathcal{U}, \mathcal{V} \in \mathcal{P}_q(n)$. It is called the {\em subspace distance} between $\mathcal{U}$ and $\mathcal{V}$. 

 A {\em subspace code} is defined to be a subset $\mathcal{C} \subseteq \mathcal{P}_q(n)$ with at least two elements. In this context, the minimum distance of a subspace code $\mathcal{C}$ is the value 
 \[
 d_S(\mathcal{C})=min\{d_S(\cU,\cV) \ | \ \cU, \cV \in \mathcal{C}, \, \cU\neq \cV\}.
 \]   
If all the elements in $\mathcal{C}$ have the same dimension, say $k$, with $1 \leq k <n$, the code $\mathcal{C}$ is called a {\em constant dimension code}. In this case $\mathcal{C}$ is a code in the Grassmannian $\mathcal{G}_q(k,n)$, that is, the set of all $k$-dimensional vector subspaces of $\bbF_q^n$. The reader is referred to \cite{KoetKschi08} for the basic background on subspace codes.
 
 As in classical Coding Theory, one of the main problems when working with subspace codes is the search for optimal codes with the largest size given a minimum distance or optimal codes with the largest minimum distance given a size.
 Let us recall some important concepts related with this problem which will be used in the sequel.


%
%

\subsection{Equidistant codes, partial spreads and spreads}\label{subsec:equidistant}


A constant dimension code $\cC \subset \mathcal{G}_q(k,n)$ is {\em equidistant} if the distance between any two distinct codewords is equal to a given value. Hence, it is satisfied that $d_S(\cC)=d_S(\mathcal{U}, \mathcal{V})$ for all $\mathcal{U}, \mathcal{V} \in \mathcal{C}$ with $\mathcal{U}\neq \mathcal{V}$. In particular, the intersection between any two different codewords of $\cC$ has a fixed dimension $c$ with $d_S(\cC)=2(k-c)$. In this case, we say that $\mathcal{C}$ is an {\em equidistant $c$-intersecting} constant dimension code. Note that the condition $n \geq 2k-c$ is necessary for the existence of that codes. Equidistant subspace codes in the Grassmannian were introduced for the first time in \cite{EtzRa15}. In this paper important examples of families of equidistant codes are described. In \cite{GoRava15} the authors give an almost complete classification of such codes in ground fields with large cardinality. 

A code $\cC\subseteq \mathcal{G}_q(k,n)$ with $d_S(\cC)=2k$ is called a {\em partial spread} code. In particular, a partial spread code is an equidistant $0$-intersecting constant dimension code ($n\geq 2k$). A systematic construction of partial spreads with efficient decoding algorithms can be found in \cite{GoRava14}. 

Note that a partial spread code attains the maximum subspace distance. Concerning the size of such codes, in  \cite{GoRava14} we can find the following result: if $\cC\subseteq \mathcal{G}_q(k,n)$ is a partial spread then 
\begin{equation}\label{eq:cota partial spread}
|\cC|\leq \left\lfloor\frac{q^n-1}{q^{k}-1}\right\rfloor.
\end{equation}
When $k$ divides $n$ this upper bound is always attained by codes that are called {\em spread codes} (or {\em $k$-spread codes}). It follows that $k$-spread codes are optimal codes with minimum distance $2k$. See \cite{MangaGorlaRosen08} and references inside for more details on spreads. 



%
%

\subsection{Multishot Codes}\label{subsec:multishot and flag}

The codes described in the previous section can be considered as examples of {\em one-shot subspace codes}, since they use the subspace channel just once. In contrast, the subspace channel can be used many times giving rise to the so-called {\em multishot subspace codes} introduced in \cite{NobUcho09}. This kind of codes appears as an alternative to one-shot subspace coding, specially when it is not possible to modify neither the field size $q$ nor the packet size $n$. Multishot subspace coding introduces a new interesting parameter in order to find codes with good rates: the number of channel uses. Moreover, as pointed out in \cite{NobUcho10}, we can obtain codes with better error-correction capabilities by spreading redundancy across multiple shots.

 A {\em multishot subspace code of length $r$} (also called an {\em $r$-shot subspace code}) over $\mathcal{P}_q(n)$, is a non-empty subset of $\mathcal{P}_q(n)^r$, that is, the $r$-th Cartesian power of the projective space. The {\em extended subspace distance} between two elements ${\mathcal{U}}=(\mathcal{U}_1,\ldots,  \mathcal{U}_r)$ and $\mathcal{V}=(\mathcal{V}_1,\ldots,  \mathcal{V}_r)$ of $\mathcal{P}_q(n)^r$,  is defined by
\begin{equation}\label{eq:distext}
d_S(\mathcal{U}, \mathcal{V})= \sum_{i=1}^r d_S(\mathcal{U}_i, \mathcal{V}_i).
\end{equation}

In \cite{NobUcho09}, the subspace dimension at each shot is unfixed and no relationship with previous shots is imposed.  However, creating dependencies among the transmitted codewords of different shots can improve the error-correction capabilities (see \cite{NobUcho10}). In this paper we explore the use of {\em multishot constant dimension codes} given by {\em nested subspaces (flags)}, that is, at each shot the dimension of the transmitted subspace is fixed and it must contain the subspace sent at the previous shot. Let us precise this idea in the following section. 


\section{On flag codes}\label{sec:on flag codes}

A {\em flag} of type $(t_1, \ldots, t_r)$, with $0<t_1<\cdots <t_r<n$, on the vector space $\mathbb{F}_q^n$ is an element $\mathcal{F}=(\mathcal{F}_1,\ldots,  \mathcal{F}_r)$ of $\mathcal{G}_q(t_1,n) \times \cdots \times \mathcal{G}_q(t_r,n) \subseteq \mathcal{P}_q(n)^r$ such that

\[
0\subsetneq \mathcal{F}_1 \subsetneq \cdots \subsetneq \mathcal{F}_r \subsetneq \mathbb{F}_q^n
\]
and $\dim(\mathcal{F}_i)=t_i$, for all $i=1, \ldots, r$. In case  the type vector is $(1, 2, \ldots, n-1)$ we say that ${\cF}$ is a {\em full flag}. Given a flag  $\mathcal{F}=(\mathcal{F}_1,\ldots,  \mathcal{F}_r)$ of type $(t_1, \ldots, t_r)$, we say that $\mathcal{F}_i$ is its {\em $i$-th subspace}.
The space of flags of type $(t_1, \ldots, t_r)$ on $\mathbb{F}_q^n$ is denoted by $\mathcal{F}_q((t_1,\ldots, t_r),n)$. In this context we can give the following definition:

\begin{definition}
A {\em flag code of type $(t_1,\ldots,t_r)$} on the vector space $\bbF_q^n$ is a subset ${\cal C}\subseteq \mathcal{F}_q((t_1,\ldots, t_r),n)$ with $|\cC|\geq 2$. 
\end{definition}

 As  a subset of $\mathcal{P}_q(n)^r$, the space $\mathcal{F}_q((t_1,\ldots, t_r),n)$ can be naturally endowed with the extended subspace distance given in (\ref{eq:distext}). We will denote it by $d_f$ in the flag codes setting.

\begin{definition}\label{distflags}
Given  $\cF=(\mathcal{F}_1,\ldots,  \mathcal{F}_r)$ and $\cF'=(\mathcal{F}'_1,\ldots,  \mathcal{F}'_r)$ two flags in $\mathcal{F}_q( (t_1, \ldots, t_r),n)$, the {\em flag distance} between $\cF$ and $\cF'$ is
\[
d_f(\cF,\cF')= \sum_{i=1}^r d_S(\mathcal{F}_i, \mathcal{F}'_i),
\]
where $d_S$ denotes the subspace distance. 
The {\em minimum distance} of a flag code $\cC$ of type  $(t_1, \ldots, t_r)$ is given by
\[
d_f(\cC)=\min\{d_f(\cF,\cF')\ |\ \cF,\cF'\in\cC, \ \cF\neq \cF'\}.
\]
\end{definition}

\begin{remark}
Observe that a flag code of type $(t_1, \ldots, t_r)$ is, in particular, a multishot constant dimension code of length $r$. Besides, a flag of type $(t_1)$ on $\mathbb{F}_q^n$ is just a vector space of dimension $t_1$ of $\mathbb{F}_q^n$ and the flag space $\mathcal{F}_q((t_1),n)$ coincides with the Grassmannian $\mathcal{G}_q(t_1,n)$. In this sense, flag codes generalize subspace codes and the flag distance is also a generalization of the subspace distance defined over the Grassmannian. 
\end{remark}

Just as constant dimension codes in $\mathcal{G}_q(k,n)$ have minimum distance upper-bounded by the value $\min\left\lbrace2k, 2(n-k)\right\rbrace$, we can also give an upper bound for flag codes of type  $(t_1, \ldots, t_r)$. To do so, take into account that given a pair of flags $\cF=(\mathcal{F}_1,\ldots,  \mathcal{F}_r)$ and $\cF'=(\mathcal{F}'_1,\ldots,  \mathcal{F}'_r)$, both of type  $(t_1, \ldots, t_r)$, for each $i\in\{1,\ldots,r\}$, it holds that

\begin{equation}\label{eq:2ti}
d_S(\cF_i,\cF'_i) \leq 2t_i , \mbox{ if } 2t_i\leq n 
\end{equation}
and 
\begin{equation}\label{eq:2(n-ti)}
d_S(\cF_i,\cF'_i) \leq 2(n-t_i), \mbox{ if } 2t_i>n.
\end{equation}

Next result follows straightforwardly. 

\begin{lemma}\label{lem:Distancia maxima flags}
Given a flag code $\cC$ ot type $(t_1,\ldots,t_r)$ on $\bbF_q^n$, we have that 
\begin{equation}\label{eq:quotamaxdistflag}
d_f(\mathcal{C}) \leq 2\left(\sum_{t_i \leq \lfloor \frac{n}{2}\rfloor} t_i + \sum_{t_i > \lfloor \frac{n}{2}\rfloor} (n-t_i)\right).
\end{equation}
In particular, when $\cC$ is a full flag code, $(\ref{eq:quotamaxdistflag})$ becomes 
\[
d_f(\mathcal{C})\leq 
\left\lbrace
\begin{array}{lccc}
\dfrac{n^2}{2}, & \text{for} & n & \text{even}, \\ 
                &            &   &         \\
\dfrac{n^2-1}{2}, & \text{for} & n & \text{odd}.
\end{array}
\right.
\]
\end{lemma}

\subsection{Flag codes with maximum distance}

%
We are interested in the family of flag codes on $\bbF_q^n$ that attain the bound given in  (\ref{eq:quotamaxdistflag}).

\begin{definition}
We say that a flag code $\cC$ ot type $(t_1,\ldots,t_r)$ is an {\em optimum distance flag  code}  if it attains the maximum possible distance for flag codes of type $(t_1, \ldots, t_r)$ given by (\ref{eq:quotamaxdistflag}).
\end{definition}

To deepen in the study of optimum distance flag codes we consider some special constant dimension codes that can be associated to a given flag code in a natural way.

\begin{definition}\label{def:Ci}
Consider  a flag code $\cC$ of type $(t_1,\ldots,t_r)$ and take an index $i \in \{1,..,r\}$. We call the {\em $i$-projected code of $\mathcal{C}$}  to the subspace code $\mathcal{C}_i$ given by the set of all the $i$-th subspaces of flags in $\mathcal{C}$. More precisely,
\[
\mathcal{C}_i= \left\lbrace \mathcal{V} \in \mathcal{G}_q(t_i,n)  \ | \  \mathcal{V}=\cF_i \ \text{for some} \ \cF=(\cF_1,\ldots,\cF_r) \in \mathcal{C}\right\rbrace. 
\]
\end{definition}

For each $i\in\{1, \ldots, r\}$, we have that $\mathcal{C}_i$ is a constant dimension code in $\mathcal{G}_q(t_i,n)$ of cardinality $\vert \cC_i\vert \leq \vert \cC \vert$. Notice that it is satisfied that  $d_S(\cC_i)\geq 0$ and, $d_S(\cC_i)=0$ if and only if $|\cC_i|=1$. 

 Our aim is to determine if a given flag code is an optimum distance flag code in terms of properties of its projected codes. In particular, an optimum distance flag code of type $(t_1)$ is a constant dimension code in the Grassmannian $\cG_q(t_1,n)$ with maximum subspace distance. So, optimum distance flag codes also generalize maximum distance constant dimension codes.

\begin{proposition}\label{prop:MFDprojected}
Let $\cC$ be an optimum distance flag code of type $(t_1, \ldots, t_r)$ on $\bbF_q^n$. Then all its projected codes attain their maximum possible subspace distance, that is, $d_S(\cC_i)=\min\left\lbrace 2t_i, 2(n-t_i)\right\rbrace.$
\end{proposition}
\begin{proof}
Assume that $\cC$ is an optimum distance flag code such that the $i$-projected code $\cC_i$ has subspace distance $d_S(\cC_i) < \min\left\lbrace 2t_i, 2(n-t_i)\right\rbrace$ for some index $1\leq i \leq r$. Then, there exist different flags $\cF, \cF' \in \cC$ such that the value $d_S(\cF_i, \cF'_i)$ does not attain the bounds (\ref{eq:2ti}) or (\ref{eq:2(n-ti)}) and, consequently, the bound given in (\ref{eq:quotamaxdistflag}) cannot be attained.
\end{proof}
\begin{remark}
The previous result states which kind of constant dimension codes can play an important role to provide families of optimum distance flag codes from its projected codes: partial spreads of dimension up to $\lfloor \frac{n}{2}\rfloor$ or equidistant codes with minimum possible subspace intersection for higher dimensions. In other words, if $\cC$ is an optimum distance flag code of type $(t_1, \ldots, t_r)$ and $\cC_i$ is its $i$-projected code, then $\cC_i$ must be a $c_i$-intersecting constant dimension code of dimension $t_i$, where $c_i=\max\left\lbrace 0, 2t_i-n\right\rbrace$.
\end{remark}

Notice that, in general, the converse of Proposition \ref{prop:MFDprojected} is not true: a flag code having maximum distance constant dimension codes as projected codes does not have to be necessarily an optimum distance flag code. The following example reflects this situation.

\begin{example}\label{ex:Fq5}
Consider the standard basis $\{e_1,...,e_5\}$ of the vector space $\bbF_q^5$. Let $\cC$ be the flag code of type $(1,3)$ on $\bbF_q^5$ consisting of the flags:
$$
\begin{array}{cccl}
\mathcal{F}^1 &=& (\left\langle e_1 \right\rangle, & \left\langle e_1, e_2, e_3 \right\rangle),\\
\mathcal{F}^2 &=& (\left\langle e_4 \right\rangle, & \left\langle e_1, e_4, e_5 \right\rangle),\\
\mathcal{F}^3 &=& (\left\langle e_1 \right\rangle, & \left\langle e_1, e_2+e_4, e_3+e_5 \right\rangle).\\
\end{array}
$$
This flag code has two projected codes with maximum subspace distance: a partial spread $\cC_1 \in \cG_q(1, 5)$ and a $1$-intersecting code $\cC_2\subset \cG_q(3,5)$. Nevertheless, the distance of $\cC$ is $d_f(\cC)=d_f(\cF^1, \cF^3)=4$ while, by means of (\ref{eq:quotamaxdistflag}),  the distance of any optimum distance flag code of type $(1,3)$ on $\bbF_q^5$ has to be equal to 6.
\end{example}

Observe that in Example \ref{ex:Fq5} the minimum distance of the flag code $\cC$ is attained between two flags with the same first subspace. Hence, despite of having a flag code $\cC$ with maximum distance projected codes, if two flags in $\cC$ share a subspace, the code $\cC$ cannot have optimum distance. At this point, we introduce a family of flag codes in which different flags have no common subspaces.

\begin{definition}
Given a flag code $\cC$ ot type $(t_1, \ldots, t_r)$, we say that $\cC$ is  a {\em disjoint flag code} if  $\vert \mathcal{C}_1\vert =\cdots=\vert\mathcal{C}_r\vert=\vert\mathcal{C}\vert,$ where $\mathcal{C}_1,...,\mathcal{C}_r$ are the projected subspace codes of $\mathcal{C}$.
\end{definition}

Next result gives a characterization of optimum distance flag codes.

\begin{theorem}\label{theo:equidistant MFD}
Let $\cC$ be a flag code of type $(t_1, \ldots, t_r)$. The following statements are equivalent:
\begin{enumerate}
\item[(i)] The code $\cC$ is an optimum distance flag code.
\item[(ii)] The code $\cC$ is disjoint and every projected code $\cC_i$ attains the maximum possible subspace distance.
\end{enumerate}
\end{theorem}
\begin{proof}
$(i)\Rightarrow (ii).$ Let $\cC$ be an optimum distance flag code. By Proposition \ref{prop:MFDprojected}, every projected code of $\cC$ attains the maximum possible subspace distance. Assume that $\cC$ is not disjoint. Then, there exists some index $j \in\{1,\ldots, r\}$ with  $\vert \mathcal{C}_j\vert < \vert \mathcal{C}\vert$ and at least two different flags $\mathcal{F}, \mathcal{F}' \in \mathcal{C}$ such that $\mathcal{F}_j=\mathcal{F}'_j$. Thus
\[
d_f(\mathcal{C})\leq d_f(\mathcal{F},\mathcal{F}') = \sum_{i\neq j} d_S(\mathcal{F}_i,\mathcal{F}'_i).
\]
By means of the bounds provided in ($\ref{eq:2ti}$) and ($\ref{eq:2(n-ti)}$), we have that $d_f(\cC)$ cannot be the maximum distance and then, the flag code $\cC$ cannot be an optimum distance flag code.

$(ii)\Rightarrow (i)$. Assume that $(ii)$ is true. Since $\cC$ is disjoint, given any pair of different flags $\cF$ and $\cF'$ in $\cC$,  for every $i=1, \ldots, r$, the subspaces $\cF_i$ and $\cF'_i$ are different and $d_S(\cF_i, \cF'_i)=d_S(\cC_i)$ is the maximum possible distance between $t_i$-dimensional subspaces. Hence, $d_f(\cC)$ attains the upper bound given in  (\ref{eq:quotamaxdistflag}) and $\cC$ is an optimum distance flag code.
\end{proof}

Taking into account this result, in the following section we propose a full flag code construction that provides a family of optimum distance flag codes with the largest possible size.


 \section{Optimum distance flag codes from spreads}\label{sec:MFDfromSpreads}

In Section \ref{sec:on flag codes} we have proved that the projected codes of an optimum distance flag code $\cC$ of type $(t_1,...,t_r)$  on $\bbF_q^n$ have to be maximum distance constant dimension codes. In particular, projected codes of dimension up to $\lfloor \frac{n}{2}\rfloor$ must be partial spreads. Moreover, if there is any dimension $t_i$ that divides $n$, then $t_i \leq \lfloor \frac{n}{2}\rfloor$, and we have that $\cC_i$ has to be a partial spread.

Recall that spread codes are codes with optimal cardinality for maximal error correction capability, that is, optimal partial spreads. Furthermore, $t_i$-spreads of $\bbF_q^n$ always exist whenever $t_i$ is a divisor of $n$ (see \cite{MangaGorlaRosen08},\cite{MangaTraut14}).

In the search of families of optimum distance flag codes with the largest possible size, it is quite natural to look for optimum distance flag codes having a spread as $i$-projected code for $t_i$ a divisor of $n$. We begin this section focusing on this question for optimum distance full flag codes. 





\bigskip

By Theorem \ref{theo:equidistant MFD}, any optimum distance flag code has to be disjoint. Then, by cardinality, we can have one spread code at most among its projected codes. In the following result we determine which dimensions $i\in \{1,\ldots, n-1\}$ could be admissible to get $\cC$ an optimum distance full flag code such that $\cC_i$ is a spread code.

\begin{proposition}\label{prop:MFD_full_flag_implies_n=2k}
Let $\cC$ be an optimum distance full flag code on $\bbF_q^n$ such that $n=ks\geq 2$. Assume that its $k$-projected code $\cC_k$ is a $k$-spread.  Then, if $n\neq 3$, we have that $s=2$.
\end{proposition}
\begin{proof}
If $n=2$, the result follows straightforwardly since a full flag code on $\bbF_q^2$ is just a subspace code in $\cG_q(1,2)$. 

Now, assume that $n=ks\geq 4$. Since $k<n$, it follows that $s>1$. Arguing by contradiction, suppose $s>2$. Let us see that, in this case, $k+1 \leq \lfloor \frac{n}{2}\rfloor$. Note that this is equivalent to show that $2(k+1) = 2k+2 \leq ks=n$, that is, $ (s-2)k\geq 2$. In case $k=1$, we have $(s-2)k=s-2=n-2 \geq 2,$ since $n\geq 4$. If $k\geq 2$, since $s-2\geq 1$ we also obtain $ (s-2)k\geq 2$.

Now, as $\cC$ is an optimum distance full flag code, by means of Theorem \ref{theo:equidistant MFD}, its $(k+1)$-projected code must be a partial spread of $\bbF_q^n$ with cardinal $\vert \cC_{k+1} \vert =\vert \cC_{k} \vert = \frac{q^n-1}{q^k-1}$. This is a contradiction, since the size of any partial spread of dimension $k+1$ is upper bounded by $\lfloor\frac{q^n-1}{q^{k+1}-1}\rfloor  < \frac{q^n-1}{q^k-1}$ by (\ref{eq:cota partial spread}).
\end{proof}

%
%
%

As a consequence of Proposition \ref{prop:MFD_full_flag_implies_n=2k},  for $n\neq 3$, if there exist optimum distance full flag codes on $\bbF_q^n$ with a $k$-spread as a $k$-projected code, the dimension $n$ must be equal to $2k$. In the following subsection we give the precise construction of an optimum distance full flag code $\cC$ on $\bbF_q^{2k}$ such that $\cC_k$ is a $k$-spread.

\begin{remark}
Notice that for $n=3$ the only admissible situation for conditions of Proposition $\ref{prop:MFD_full_flag_implies_n=2k}$ corresponds to $k=1$, since $k<n$. In this case, full flags consist of nested lines and planes. This is a particular case of a much more general study that we will address in a forthcoming paper (see Section $5$).

\end{remark}

\subsection{Our Construction}\label{subsec:construction}
Throughout this section, we will assume that $n=2k$ and $k\geq 2$. Let $\mathcal{S}$ be a {\em planar spread} of $\mathbb{F}_q^{2k}$, that is, a spread $\cS \subseteq \mathcal{G}_q(k,2k)$. In particular, we have that $\cS$ has cardinality $|\cS|=q^k+1$ and we can write  $\mathcal{S}=\{\mathcal{S}_1,\ldots, \mathcal{S}_{q^k+1}\}$. The subspaces $\cS_i$ satisfy 
\begin{equation}\label{eq:planars}
\mathcal{S}_i\cap \mathcal{S}_j=\{0\}  \mbox{ and } \mathcal{S}_i + \mathcal{S}_j=\mathbb{F}_q^{2k}, 
\end{equation}
whenever $i\neq j$.

Now, for every $i \in \{1,...,q^k+1\}$, we consider $\mathrm{S}_i$ a {\em generator matrix} of the $k$ dimensional subspace $\mathcal{S}_i$. That means that $\mathrm{S}_i$ is a full-rank $(k\times 2k)$-matrix such that $\mathcal{S}_i=\mathrm{rowsp}(\mathrm{S}_i)$. Fixed an index $j\leq k$, we will denote by $\mathrm{S}_i^{(j)}$ the submatrix of $\mathrm{S}_i$ given by the first $j$ rows of $\mathrm{S}_i$. 


Since $\cC$ is a planar spread, any matrix of the form $\begin{pmatrix} 
 \mathrm{S}_{i_1}\\
 \mathrm{S}_{i_2}
\end{pmatrix}$ is a $2k \times 2k$ full rank matrix for $i_1 \neq i_2$. In particular, if we denote 
\begin{equation}\label{eq:matricesW_i}
 \mathrm{W}_i=\left(\begin{array}{lc}
 \mathrm{S}_{i}\\
 \mathrm{S}_{i+1}
\end{array}\right) \textrm{ for } i=1,\ldots ,q^k \textrm{ and }  \mathrm{W}_{q^k+1}=\left(\begin{array}{lc} 
 \mathrm{S}_{q^k+1}\\
 \mathrm{S}_{1}
\end{array}\right),
\end{equation}
we have that all the matrices $ \mathrm{W}_i$ are $2k \times 2k$ full rank matrices. Moreover, $ \mathrm{W}_i^{(k)}=S_i$ for any $i=1,...,q^k+1$. We will denote by 
\[
\cW_i^{(j)}= \mathrm{rowsp}(\mathrm{W}_i^{(j)})
\]
the subspace of dimension $j$ generated by the first $j$ rows of $ \mathrm{W}_i$.
Notice that, given $i \in\{1,...,q^k+1\}$ and $1\leq j_1<j_2\leq 2k$, it holds that
$
\cW_i^{(j_1)}\subsetneq \cW_i^{(j_2)}.
$
As a consequence, given a matrix $\mathrm{W}_i$ as above, we can define the {\em full flag associated to $ \mathrm{W}_i$} as
\begin{equation}\label{eq:flag generated}
\cF_{ \mathrm{W}_i}=(\cW_i^{(1)},\ldots,\cW_i^{(2k-1)}).
\end{equation}
Finally, we define the {\em full flag code associated to the matrices $\{ \mathrm{W}_i\}_{i=1}^{q^k+1}$} as
\begin{equation}\label{eq:our code}
\cC=\{\cF_{ \mathrm{W}_i}\ |\ 1\leq i\leq q^k+1\}.
\end{equation}
It follows that the $j$-projected code of $\cC$ is 
\begin{equation}\label{eq:our projected code}
\cC_j=\{\cW_i^{(j)}\ |\ 1\leq i\leq q^k+1\}.
\end{equation}

Our aim is to prove that the code $\cC$ defined in (\ref{eq:our code}) is an optimum distance full flag code  such that $\cC_k=\cS$ (see Theorem \ref{Teorema flag-code}). To do this, we proceed in two steps: first we show that $\cC_j$ is a partial spread, for each $1\leq j\leq k$ and, secondly, we will show that $\cC_j$ is an equidistant $2(j-k)$-intersecting code, for each $k<j<2k$.




\begin{proposition}\label{prop:construccion 1y2 juntas}
Let $\mathcal{S}$ be a $k$-spread of $\mathbb{F}_q^{2k}$. Using the previous notation, for all $1\leq j\leq 2k-1$, the set 
\[
\mathcal{C}_j=\left\lbrace \cW_i^{(j)}\ |\  \ i=1, \ldots, q^k+1\right\rbrace,
\]
is a constant dimension code of $\mathcal{G}_q(j,2k)$ with cardinality $\vert\mathcal{C}_j\vert=\vert \mathcal{S}\vert$. Moreover,
\begin{enumerate}
  \item  $\cC_j$ is a partial spread, for $1\leq j\leq k$, and
  \item $\cC_j$ is equidistant $2(j-k)$-intersecting, for $k<j\leq 2k-1$. So, its distance is $2(2k-j)$, that is the maximum possible distance between a pair of $j$-dimensional subspaces in $\bbF_q^{2k}$.
\end{enumerate}
\end{proposition}

\begin{proof} 
Since the matrices  $\{ \mathrm{W}_i\}_{i=1}^{q^k+1}$ defined in (\ref{eq:matricesW_i}) have full rank, the first $j$ rows of  $\mathrm{W}_i$ are linearly independent for all $1\leq j\leq 2k-1$. Hence, $\cW_i^{(j)}$ is a $j$-dimensional subspace of $\bbF_q^{2k}$ and $\mathcal{C}_j$ is a constant dimension code of $\mathcal{G}_q(j,2k)$, for every $1\leq j\leq 2k-1$.

Let us consider two different subspaces $\cW_i^{(j)}$  and $\cW_l^{(j)}$ in $\cC_j$, with $i, l \in\{1,\ldots ,q^k+1\}$, $i\neq l$.
For $1\leq j\leq k$,  it holds that 
$$
\cW_i^{(j)} \cap \cW_l^{(j)}=\mathrm{rowsp}(\mathrm{S}_i^{(j)}) \cap \mathrm{rowsp}(\mathrm{S}_l^{(j)})\subseteq \mathcal{S}_i \cap \mathcal{S}_l= \left\lbrace 0 \right\rbrace. 
$$
As a consequence,  we have $d_S(\mathcal{C}_j)=2j$ and $\cC_j$ is a partial spread satisfying $\vert\mathcal{C}_j\vert =\vert \mathcal{S}\vert$.
Now, for $k<j\leq 2k-1$, the sum subspace of $\cW_i^{(j)}$  and $\cW_l^{(j)}$ has dimension
\[
\dim\left( \mathcal{W}_i^{(j)} + \mathcal{W}_l^{(j)}\right)=
\mathrm{rk}
\left(
\begin{array}{l}
\mathrm{S}_i\\
\mathrm{S}_{i+1}^{(j-k)}\\
\mathrm{S}_l\\
\mathrm{S}_{l+1}^{(j-k)}
\end{array}
\right)
=\mathrm{rk}\begin{pmatrix}
\mathrm{S}_i\\
\mathrm{S}_l
\end{pmatrix}
=2k.
\]
Besides, we have that
$$
\dim\left( \mathcal{W}_i^{(j)} + \mathcal{W}_l^{(j)}\right)= 2j- \dim \left( \mathcal{W}_i^{(j)} \cap \mathcal{W}_l^{(j)}\right).
$$
As a consequence, one has that $\dim \left( \mathcal{W}_i^{(j)} \cap \mathcal{W}_l^{(j)}\right)=2(j-k)$. 
We conclude that $\mathcal{C}_{j}$ is an equidistant $2(j-k)$-intersecting code with subspace distance  $d_S(\mathcal{C}_{j})=2(j-2(j-k))=2(2k-j)$ and size $\vert\mathcal{C}_j\vert =\vert \mathcal{S}\vert$.

\end{proof}

From the  previous proposition we can directly conclude the main result in this section:

\begin{theorem}\label{Teorema flag-code}
Let $\mathcal{S}$ be a $k$-spread of $\mathbb{F}_q^{2k}$ with generator matrices $\{ \mathrm{S_i}\}_{i=1}^{q^k+1}$. Consider the matrices $\{ \mathrm{W}_i\}_{i=1}^{q^k+1}$ defined in $(\ref{eq:matricesW_i})$. Then, the set of full flags associated to these matrices, i.e.,  
\[
\mathcal{C}=\{\cF_{ \mathrm{W}_i}\ |\  i=1,\ldots,q^k+1 \},
\] 
is an optimum distance full flag code with distance $2k^2$ and size $|\mathcal{C}|=|\mathcal{S}|=q^k+1$.
\end{theorem}

\begin{proof}
By  Proposition \ref{prop:construccion 1y2 juntas}, we have that the projected subspaces of $\mathcal{C}$ satisfy the following:
\begin{enumerate}
\item $\mathcal{C}_j$ is a partial spread of $\mathcal{G}_q(j,2k)$, for $j=1,\ldots, k-1$,
\item $\mathcal{C}_k=\mathcal{S}$,
\item $\mathcal{C}_{j}$ is an equidistant $2(j-k)$-intersecting code of $\mathcal{G}_q(j,2k)$, for $j=k+1,\ldots, 2k-1$.
\item $\vert\mathcal{C}_1\vert = \cdots = \vert\mathcal{C}_k\vert = \cdots = \vert\mathcal{C}_{2k-1}\vert = \vert\mathcal{C}\vert$.
\end{enumerate}
Therefore, $\mathcal{C}$ is a disjoint flag code such that every projected code attains the maximum possible subspace distance. By means of Proposition \ref{theo:equidistant MFD}, we conclude that $\mathcal{C}$  is an optimum distance flag code. Since $\cC$ is a full flag code, by Lemma \ref{lem:Distancia maxima flags}, it follows that $d_f(\cC)=2k^2$.
\end{proof}

Once we have constructed the optimum distance full flag code given in Theorem \ref{Teorema flag-code}, we wonder if there exist other optimum distance full flag codes on $\bbF_q^{2k}$, not necessarily having a spread as $k$-projected code, with cardinality higher than $q^k+1$. Next, we show that this cardinality cannot be improved and hence, the flag code given by our construction has also the maximum possible size among optimum distance full flag codes on $\bbF_q^{2k}.$

\begin{theorem}\label{Teorema Maximalidad}
Let $\cC$ be an optimum distance full flag code on $\bbF_q^{2k}$. Then  $|\cC| \leq q^k+1$. The equality holds if, and only if, the $k$-projected code of $\cC$ is a planar spread of $\bbF_q^{2k}.$
\end{theorem}
\begin{proof}
Let ${\mathcal{C}}$ be an optimum distance full flag code on $\mathbb{F}_q^{2k}$ and consider its $k$-projected code ${\mathcal{C}}_k \subset \cG_q(k,2k)$. 
By Theorem \ref{theo:equidistant MFD}, we know that ${\cC}$ is a disjoint flag code and its projected codes have maximum distances. In particular, $\vert {\cC} \vert = \vert {\cC}_k \vert$ and ${\cC}_k$ is a partial spread of $\bbF_q^{2k}$. Hence, $\vert {\cC} \vert = \vert {\cC}_k \vert \leq \frac{q^{2k}-1}{q^k-1}=q^k+1$ and, the equality holds only if ${\cC}_k$ is a spread code.
\end{proof}

To present an example of optimum distance full flag code on $\bbF_q^{2k}$ having a $k$-spread as a $k$-projected code, it is enough to choose a family of generator matrices of a given planar spread. We use the spread constructed in  \cite{MangaGorlaRosen08}.
\begin{example}
Let $\mathrm{M}$ denote the companion matrix of a monic, primitive polynomial of degree $k$ and coefficients on $\bbF_q$. The set $\mathcal{S}\subset \mathcal{G}_q(k,2k)$, with elements generated by matrices
\[
\begin{array}{ccl}
 \mathrm{S}_i & = & \left[\mathrm{I}_k \vert \mathrm{M}^i\right], \ i= 1, \ldots, o(\mathrm{M})=q^k-1,\\
 \mathrm{S}_{q^k} & = & \left[\mathrm{I}_k \vert \,  \mathrm{0}\,  \right],\\
 \mathrm{S}_{q^k+1} & = & \left[\, \mathrm{0} \, \vert \mathrm{I}_k\right],
\end{array}
\]
is a planar spread of $\bbF_q^{2k}$ (see \cite{MangaGorlaRosen08}). Now, following the construction provided by Theorem \ref{Teorema flag-code},  from the planar spread $\mathcal{S}$ we can obtain an optimum distance full flag code of cardinality $q^k+1$.
\end{example}

\begin{remark}
In \cite{LiebNebeVaz18} the authors present several constructions of flag codes as orbits of group actions. None of that constructions attain the bound given in (\ref{eq:quotamaxdistflag}) and therefore none of them are optimum distance flag codes. 
\end{remark}


\subsection{A decoding process on the erasure channel} \label{subsec:decoding}

We fix $\mathcal{C}$ a full flag code from a given planar spread defined as in Theorem \ref{Teorema flag-code} as our error correcting flag code and we propose a decoding algorithm on the erasure channel. 
In the general setting of error correcting models based on minimum distance, it is well known that any code  with minimum distance $d$ can detect up to $d-1$ errors and can correct, at most, $\lfloor \frac{d-1}{2}\rfloor$ of them. In particular, as stated in Theorem \ref{Teorema flag-code}, our flag code $\mathcal{C}$ has flag distance $d_f(\mathcal{C})=2k^2$, so it can detect up to $2k^2-1$ errors and correct at most $\left\lfloor\frac{2k^2-1}{2}\right\rfloor=k^2-1$ errors. Moreover, if we consider its projected codes $\mathcal{C}_i$ as independent subspace codes, their respective error-correction capabilities are $\left\lfloor \frac{d_S(\mathcal{C}_i)-1}{2}\right\rfloor=\frac{d_S(\mathcal{C}_i)}{2}-1$ for $i=1,\ldots, 2k-1$.

Now, assume we have sent a flag $\cF=(\cF_{1},\ldots, \cF_{2k-1}) \in \cC$ and the receiver has obtained a sequence of subspaces $\cX= (\cX_1, \ldots , \cX_{2k-1})$. We denote the {\em total error} of the communication by 
\[
e=d_f(\cF,\cX)=\sum_{i=1}^{2k-1} d_S(\mathcal{F}_i, \cX_i).
\]
Moreover,  the {\em $i$-th (shot) subspace error}  will be denoted by $e_i=d_S(\mathcal{F}_i, \cX_i)$. 

As $\cC$ is an optimum distance flag code, in particular $d_{f}(\cC)= \sum_{i=1}^{2k-1}d_S(\mathcal{C}_i)$ by Theorem \ref{theo:equidistant MFD}. Based on this property, in the following result we prove that if the total error $e$ is correctable by our flag code $\cC$, that is, $e\leq k^2-1$, at least one of the associated subspace errors must be also correctable. As a consequence, at least one of the received subspaces could be decoded by minimum distance in the corresponding projected code.

\begin{proposition}\label{AlMenosUnoDecodificaBien}
Assume that $e$ is a correctable total error. Then there exists at least one index $i\in\{1,\ldots, q^{k}+1\}$ such that the $i$-th subspace error $e_i$ is correctable.
\end{proposition}
\begin{proof} Suppose that none of the subspace errors is correctable. That means that  $e_i > \frac{d_S(\mathcal{C}_i)}{2}-1$ for all $i$. Hence, we have that $e_i \geq \frac{d_S(\mathcal{C}_i)}{2}$. By Theorem \ref{theo:equidistant MFD}, if we compute the total error, we get
$$
e=\sum_{i=1}^{2k-1} e_i \geq \sum_{i=1}^{2k-1} \frac{d_S(\mathcal{C}_i)}{2}= \frac{d_f(\mathcal{C})}{2}=k^2 > k^2-1.
$$ This is a contradiction, since $e$ is correctable.
\end{proof}

Notice that Proposition \ref{AlMenosUnoDecodificaBien}  still holds true in the more general setting of multishot codes in which the extended distance of the code is the sum of the subspace distances of the codes used at each shot.  In addition to this property, the nested structure of our flag code will be useful in the decoding process as we will see next.

\subsubsection*{The erasure channel model}

Our code $\mathcal{C}$ is, in particular, a multishot subspace code \cite{NobUcho09} and we will use the subspace channel at most $2k-1$ times.  Nevertheless, since we are working with flags, we use the general idea of the channel model proposed in \cite{LiebNebeVaz18} in order to take advantage of the nested subspaces we have. As we will see later, this fact improves the error-correction capability of the code, as it was already suggested in \cite{NobUcho10}.

The network is modelled as a finite directed acyclic multigraph with a single source and possibly multiple receivers.  The source and the receivers agree in some set of flags, our code $\cC$ in this case, and the information is encoded as a flag in $\cC$. As we assume we work in an erasure channel, erasures are allowed during the transmission process but there are no errors at any step.

Suppose we want to send the full flag $\cF \in \mathcal{C}$, associated to the matrix $\mathrm{W}_s$, for some $s\in\{1,\ldots, q^{k}+1\}$. Recall that if $\cF_{i}$ is the $i$-subspace of $\cF$, according to (\ref{eq:flag generated}), we have that $\cF_{i}=\cW_s^{(i)}=\rsp(\mathrm{W}_s^{(i)})$, for $i=1, \ldots, 2k-1$.  In principle, the source sends this full flag $\cF$ in $2k-1$ shots. Nevertheless, as we will see in Propositions \ref{DecodingBeforek} and \ref{DecodingAfterk}, we usually could get $\cF$ in  a significantly fewer number of shots. 

Next we describe the general process of transmission. At the $i$-th shot:

\begin{itemize}
\item Through every outgoing edge, the source sends the $i$-th row of the generator matrix of  $\cF_{i}$, that is, $\mathrm{W}_s^{(i)}.$
\item Then, every intermediate node constructs random linear combinations of the vectors that it has received up to this point and sends each of them through an outgoing edge.
\item The receiver obtains many (say $a_i$) random linear combinations of the rows of  $\mathrm{W}_s^{(i)}$ and gets a matrix $Z_i=Y_i\mathrm{W}_s^{(i)}$, where $Y_i$ is a $(a_i \times i)$-matrix. The receiver gathers the matrices $Z_1,\ldots ,Z_i$ received until this moment and defines the subspace

 
\begin{equation}\label{eq:channel}
\cX_i=\mathrm{rowsp}
\begin{pmatrix}
Z_1\\
\vdots\\
Z_i
\end{pmatrix} \subseteq \cF_{i}.
\end{equation}

\item Finally, at the last shot, i.e., the $(2k-1)$-th shot, the sequence of nested subspaces $\cX= (\cX_1, \ldots , \cX_{2k-1})$ is received.
\end{itemize}

Notice that if for an index $i \in \{1,...,2k-1\}$, the last column of $Y_j$ is not null for all $j \leq i$, then $\cX_i=\cF_{i}$. If this happens at any shot, then  $\cX=\cF$. Otherwise, erasures have occurred during the transmission process and $\cX$ is an {\em stuttering flag}, that is, a sequence of nested subspaces  where equalities are allowed (see \cite{LiebNebeVaz18}). Nevertheless, we always have that $\cX_i\subseteq \cF_{i}$ for all $i$.

\begin{remark}
 Observe that this channel model takes advantage of the nested structure of flags in order to reduce the number of erasures that occur in the communication process. If we send a flag as a codeword of a multishot code, without regarding the nested structure of its subspaces, at the $i$-th shot, the receiver gets a matrix $Z_i=Y_i\mathrm{W}_s^{(i)}$ and constructs a subspace $\bar{\cX}_i=\rsp(Z_i).$ After $2k-1$ shots, the receiver has obtained a sequence of subspaces $\bar{\cX}=(\bar{\cX}_1, \ldots, \bar{\cX}_{2k-1})$. In that case, by construction, it holds $\bar{\cX}_i \subseteq \cX_i \subseteq \cF_i$ and then, 
$$
\begin{array}{rcl}
d_S(\cF_i, \cX_i) & = &     \dim(\cF_i + \cX_i)-\dim(\cF_i \cap  \cX_i)\\
                  & = &     \dim(\cF_i) -\dim(\cX_i) \\
                  & \leq  & \dim(\cF_i)- \dim(\bar{\cX}_i) \\
                  & = &     d_S(\cF_i, \bar{\cX}_i).
\end{array}
$$
Hence, the stuttering flag $\cX$ is always, at least, as closer to the sent flag as the sequence of subspaces $\bar{\cX},$ in spite of the fact that the same number of erasures has occurred at each shot. So, we can say that sending flags, which have a lot of redundancy in their structure, makes possible to correct some erasures during the transmission process. Even more, although erasures occur at some shot, in some cases, the channel itself can correct them and the receiver can obtain the sent flag, as in the following example.
\end{remark}

\begin{example}
Suppose we send the flag $\cF=(\langle e_1 \rangle, \langle e_1, e_2 \rangle, \langle e_1, e_2, e_3\rangle)$ of type $(1,2,3)$ on $\bbF_q^4$, where $\lbrace e_1, e_2, e_3, e_4 \rbrace$ represents the standard basis of $\bbF_q^4$. Suppose that during the communication process, some erasures have occurred and the receiver gets matrices
$$
Z_1=\begin{pmatrix}
1 & 0 & 0 & 0
\end{pmatrix}, \
Z_2=\begin{pmatrix}
0 & 0 & 0 & 0\\
0 & 1 & 0 & 0
\end{pmatrix}, \ \text{and} \
Z_3=\begin{pmatrix}
1 & 0 & 0 & 0\\
0 & 0 & 0 & 0\\
0 & 0 & 1 & 0
\end{pmatrix}.
$$
Nevertheless, the received sequence $\cX$ coincides with the sent flag $\cF$.
\end{example}

\subsubsection*{A decoding algorithm}

Suppose that we have sent a flag $\cF=(\cF_1, \ldots, \cF_{2k-1})$ belonging to our code $\cC$ and, through an erasure channel, the receiver has obtained a stuttering flag $\cX=(\cX_1, \ldots, \cX_{2k-1}),$ defined as in (\ref{eq:channel}). Whenever the total error $e=d_f(\cX, \cF)$ is correctable by our code $\cC$, we propose a decoding algorithm based on the following results.

\begin{proposition}\label{DecodingBeforek}
If there exists $i\in \{1,\ldots, k\}$ such that the corresponding received subspace $\cX_i$ is non trivial, then we decode $\cX$ into the unique $\cF\in\cC$ such that $\cX_{i}\subseteq \cF_{i}$.
\end{proposition}
\begin{proof}
Suppose that a nontrivial subspace $\cX_i$ is received for some $i\leq k$. As we are working in an erasure channel, $\cX_i$ must be contained in some subspace $\cU$ of $\mathcal{C}_i$. Moreover, since $\cC_{i}$ is a partial spread, we have that $\cU$ has to be the only subspace in $\mathcal{C}_i$ that contains $\cX_i$. Now, Theorem \ref{Teorema flag-code} states that $|\cC|=|\cC_{i}|$. Therefore, we can recover $\cF$ as the only flag in $\mathcal{C}$ such that $\cF_{i}=\cU$. 
\end{proof}

In the case not considered in the previous proposition, that is, when $\cX_1=\cdots=\cX_k=\left\lbrace 0\right\rbrace$, the following result holds.

\begin{lemma}\label{LemmaDecoding}
Assume we have received a subspace sequence $\cX= (\cX_1, \ldots , \cX_{2k-1})$ with $\cX_1=\cdots=\cX_k=\left\lbrace 0\right\rbrace$. If the total error is correctable  then, there exist an index $k < i < 2k$ such that $\dim(\cX_{i}) > 2(i-k)$.
\end{lemma}
\begin{proof}
Suppose that $\dim(\cX_{i})\leq 2(i-k)$, for all $i=k+1, \ldots, 2k-1$.  As we only allow erasures in the transmission, we know that $\cX_{i} \subseteq \mathcal{F}_{i}$. In this case, it holds:
$$
\begin{array}{ccccccccccc}
e_i     & = & d_S(\cX_i, \cF_i) & = & i &= &\frac{d_S(\mathcal{C}_i)}{2}, & \text{for}& i=1, \ldots, k,\\
e_{i} & = & d_S(\cX_{i}, \mathcal{F}_{i})& \geq & 2k-i &= & \frac{d_S(\mathcal{C}_{i})}{2}, & \text{for}& i=k+1, \ldots, 2k-1.
\end{array}
$$

Hence, for every $i$, the $i$-th subspace error exceeds the error-correction capability of the $i$-projected code and, as a consequence of  Proposition \ref{AlMenosUnoDecodificaBien}, the total error is not correctable, which is a contradiction.
\end{proof}

Lemma \ref{LemmaDecoding} helps us to state the following proposition. 

\begin{proposition}\label{DecodingAfterk}
 Assume that a sequence of subspaces $\cX= (\cX_1, \ldots , \cX_{2k-1})$  is received with $\cX_1=\cdots=\cX_k=\left\lbrace 0\right\rbrace$. Consider the minimum $i\in \{k+1,\ldots, 2k-1\}$ such that $\dim(\cX_{i})>2(i-k)$. Then, we can recover the sent flag $\cF$ as the only flag in $\cC$ such that $\cX_{i}$ is contained in $\cF_{i}.$
\end{proposition}
\begin{proof}
Let $i$ be the minimum index in $ \left\lbrace k+1, \ldots, 2k-1\right\rbrace$ such that $\dim(\cX_{i})>2(i-k)$. Recall that $\cX_{i} \subseteq \mathcal{F}_{i}$. Moreover, since $\mathcal{C}_{i}$ is an equidistant $2(i-k)$-intersecting subspace code and $\dim(\cX_{i})>2(i-k)$, no more subspace in $\cC_{i}$ than $\mathcal{F}_{i}$  can contain $\cX_{i}$. Thus, using that our code $\cC$ is disjoint, we can recover the sent flag as the only flag in $\mathcal{C}$ having $\mathcal{F}_{i}$ as a $i$-th subspace.
\end{proof}

These previous results make the following algorithm work:

\vspace{-15pt}
\begin{quote}

\hrulefill \\
\textbf{Decoding algorithm}

\vspace{-15pt}
\hrulefill

\begin{small}
\noindent\textbf{Data:} A stuttering flag $\cX=(\cX_1, \ldots, \cX_{2k-1}).$\\
\noindent\textbf{Result:} The sent flag $\cF \in \cC$.

\noindent \textbf{for} $1\leq i\leq 2k-1$
\vspace{-6pt}
\begin{itemize}
\item[]\textbf{if} $i\leq k$ and $\dim(\cX_i)>0$,
\vspace{-3pt}
\begin{itemize}
\item[]decode $\cX_i$ into the only $\cF_i \in \cC_i$ that contains $\cX_i,$
\item[]\textbf{return} the only flag $\cF \in \cC$ that has $\cF_i$ as $i$-th subspace.
\end{itemize}

\item[] \textbf{if} $i>k$ and $\dim(\cX_i) > 2(i-k),$
\vspace{-3pt}
\begin{itemize}
\item[]decode $\cX_i$ into the only $\cF_i \in \cC_i$ that contains $\cX_i,$
\item[]\textbf{return} the only flag $\cF \in \cC$ that has $\cF_i$ as $i$-th subspace.
\end{itemize}
\end{itemize}
\vspace{-18pt}
\end{small}
\hrulefill
\end{quote}


By means of Lemma $\ref{LemmaDecoding}$ one of the two previous conditions must be reached. Hence, our decoding algorithm allows the receiver to recover the sent flag in, at most, $2k-1$ uses of the channel.

\begin{remark}
Notice that our code $\cC$ does not need necessarily that the whole flag has been sent to decode it. During the transmission process, a given flag is sequentially sent and, at the $i$-th shot, the receiver gets subspaces  $\cX_1,\ldots,\cX_i$. At this moment either Propositions \ref{DecodingBeforek} or \ref{DecodingAfterk} could be applied in order to recover the sent flag. In the worst case, the receiver has to wait until the last shot. That means that, at the $(2k-1)$-th shot, it gets
$$
\cX=(0, \ldots, 0, \cX_{k+1}, \ldots, \cX_{2k-2}, \cX_{2k-1}),
$$
where $\dim(\cX_{i})\leq 2(i-k),$ for $i=k+1, \ldots, 2k-2$. In this situation, if the total error is correctable, by Lemma \ref{LemmaDecoding}, we have that $\dim(\cX_{2k-1})$ has to be $2k-1$. Thus, $\cX_{2k-1}=\cF_{2k-1}$ and, at the $(2k-1)$-th shot, we can recover $\cF$ as the only flag in $\cC$ having $\cF_{2k-1}$ as its $(2k-1)$-th subspace.
\end{remark}

%

\subsection{General type constructions}\label{subsec: general type}

We can easily derive constructions of optimum distance flag codes from the construction we provided in Theorem \ref{Teorema flag-code} for full flag codes. Next we explain these constructions.

\subsubsection*{Optimum distance flag codes of any type from planar spreads}

Our construction for optimum distance full fag codes from planar spreads of $\mathbb{F}_q^{2k}$ can be modified in order to get optimum distance flag codes of a general type vector $(t_1, \ldots, t_r)$ just by removing the projected codes of dimensions not appearing in the vector type. This procedure, that we call \emph{puncturing}, transforms any full flag $\cF=(\cF_1, \ldots, \cF_{2k-1})$ into the \emph{punctured} flag $\cF^{(t_1, \ldots, t_r)}=(\cF_{t_1},   \ldots, \cF_{t_r})$ of type $(t_1, \ldots, t_r)$. In that way, we can also define the \emph{punctured flag code} of the full flag code $\cC$ constructed in Theorem \ref{Teorema flag-code} as the set
\[
\mathcal{C}^{(t_1, \ldots, t_r)} = \lbrace \cF^{(t_1, \ldots, t_r)} \ | \ \cF \in \cC \rbrace \subseteq \cF_q((t_1, \ldots, t_r),n).  
\]
Notice that $\mathcal{C}^{(t_1, \ldots, t_r)}$ is also a disjoint flag code and its projected codes are maximum distance subspace codes, since they are also projected codes of $\cC$. Hence, by means of Theorem \ref{theo:equidistant MFD}, the punctured flag code $\mathcal{C}^{(t_1, \ldots, t_r)}$ is an optimum distance flag code of type $(t_1, \ldots, t_r)$. Besides, if dimension $k$ remains in the type vector $(t_1, \ldots, t_r)$, by arguing as in the proof of Theorem \ref{Teorema Maximalidad} we conclude that the cardinality of $\mathcal{C}^{(t_1, \ldots, t_r)}$, that is also $q^k+1$, is maximum. 



\subsubsection*{Optimum distance flag codes of type up to a divisor of $n$}

In the general universe of  $\bbF_q^{n}$, the construction of optimum distance flag codes of type $(t_1, \ldots, t_r)$, being $t_r$ a divisor of $n$, is always possible. To do so, just follow the ideas we used in the proof of the first part of Proposition \ref{prop:construccion 1y2 juntas}.  Let $\cS$ 
be a $t_r$-spread of $\bbF_q^n$ and for every $\cS_i\in\cS$ consider a $(t_r\times n)$-matrix of full rank $\mathrm{S}_i$ such that  $\cS_i= \rsp(\mathrm{S}_i)$.  For every $1 \leq j < r$, denote by $\mathrm{S}_i^{(j)}$ the matrix composed by the first $t_j$ rows of $\mathrm{S}_i$ and $\cS_i^{(j)}= \rsp(\mathrm{S}_i^{(j)})$ the corresponding vector subspace. It follows that the code
\[
\cC_j =\{ \cS_i^{(j)} \ \vert \ i=1, \ldots, \vert \cS \vert \}
\]
is a partial spread of dimension $t_j$ and cardinality $\vert \cS \vert$. Hence, by means of Theorem \ref{Teorema flag-code}, the flag code
$$
\cC =\{  (\cS_i^{(1)}, \ldots, \cS_i^{(r-1)},  \cS_i ) \ \vert \ i=1, \ldots, \vert \cS \vert \}
$$
is an optimum distance flag code of type $(t_1, \ldots, t_r)$ and has a $t_r$-spread as its last projected code.

%
%
%


\section{Conclusions and future work} \label{sec:conclusions}

In this paper, we have introduced several new concepts related to flag codes, such as projected subspace codes, disjoint flag codes and optimum distance flag codes. Besides, we have characterized optimum distance flag codes as disjoint flag codes having maximum distance constant dimension codes as projected codes.

In our search of constructions for optimum distance full flag codes, we have focused on the family of full flag codes on $\bbF_q^{n}$ having a $k$-spread as a projected code, for some divisor $k$ of $n$  and we have concluded that, except for $n=3$, these codes can only be constructed on $\bbF_q^{2k}$. For that case, we have provided a construction of optimum distance full flag codes based on the properties of planar spreads. This construction attains the maximum possible size, that is, $q^k+1$. Moreover, a decoding algorithm based on the properties of partial spreads and equidistant constant dimension codes is also given. 

In a forthcoming paper, we will  address the problem of obtaining a systematic construction of optimum distance flag codes of general type on $\bbF_q^{n}$ having a $k$-spread as a projected code, being $k$ a divisor of $n$. It would be interesting to characterize the admissible vector types  for which it is possible to generalize our model as well as to provide a specific construction. 



%
%
%
%



%
%
%

\end{document}